\documentclass[12pt,a4paper]{article}
\usepackage{amssymb}
\usepackage{amsmath,amsfonts,amsthm}

\usepackage{graphicx}
\usepackage{amsmath}

\usepackage{amsfonts}
\usepackage{amsthm,amscd}
\usepackage{graphicx}
\usepackage{amsmath}
\usepackage{amssymb}
\usepackage{amstext}

\newcommand\unit{\hbox{\rm 1\kern-2.8truept l}}
\newcommand\re{{\Re}\kern-1pt e}
\newcommand\im{{\Im}\kern-1pt m}
\newcommand\tr{\text{Tr}}

\newcommand{\C}{\mathbb{C}}

\newcommand\Lform{{\mathcal{L}}\kern-7.56pt\raise1.0pt\hbox{$-$}}

\newcommand{\Ll}{\mathcal{L}}

\newtheorem{definition}{Definition}
\newtheorem{theorem}[definition]{Theorem}
\newtheorem{proposition}[definition]{Proposition}
\newtheorem{example}{Example}
\newtheorem{remark}{Remark}
\newtheorem{lemma}[definition]{Lemma}
\newtheorem{corollary}[definition]{Corollary}

\newcommand{\bra}[1]{\langle\,#1\,\vert}
\newcommand{\ket}[1]{\vert\,#1\,\rangle}
\newcommand{\outerp}[2]{\ket{#1}\bra{#2}}

\newcommand{\D}{\text{D}}
\newcommand{\p}{\text{p}}
\newcommand{\Diag}{\text{Diag}}
\newcommand{\Id}{\text{I}}

\title{Thermodynamic formalism for continuous-time quantum Markov  semigroups: the detailed balance condition, entropy, pressure and equilibrium  quantum processes}

\author{Jader E. Brasil, Josue Knorst and Artur O. Lopes}

\begin{document}

\maketitle

\begin{abstract}

$M_n(\mathbb{C})$ denotes the set of  $n$ by $n$ complex matrices.
Consider  continuous time quantum   semigroups $\mathcal{P}_t= e^{t\, \mathcal{L}}$, $t \geq 0$, where $\mathcal{L}:M_n(\mathbb{C}) \to M_n(\mathbb{C})$ is the infinitesimal generator. If we assume that $\mathcal{L}(I)=0$, we will call  $e^{t\, \mathcal{L}}$, $t \geq 0$ a quantum Markov semigroup. Given a stationary density matrix $\rho=
\rho_{\mathcal{L}}$, for the  quantum Markov  semigroup  $\mathcal{P}_t$, $t \geq 0$,
we can define a  continuous time  stationary quantum Markov process, denoted by $X_t$, $t \geq 0.$  Under the detailed balance condition,
given an {\it a priori} Laplacian operator $\mathcal{L}_0:M_n(\mathbb{C}) \to M_n(\mathbb{C})$,
we will present a natural concept of entropy for a class of density matrices on $M_n(\mathbb{C})$. Given an Hermitian operator  $A:\mathbb{C}^n\to \mathbb{C}^n$ (which plays the role of an Hamiltonian), we will   study a version of the variational principle of pressure for $A$.  A density matrix $\rho_A$ maximizing pressure will be called an equilibrium density matrix.  From $\rho_A$ we will derive a  new  infinitesimal generator $\mathcal{L}_A$. Finally,   the continuous time quantum Markov  process defined by the semigroup $\mathcal{P}_t= e^{t\, \mathcal{L}_A}$, $t \geq 0$, and an initial stationary density matrix, will be called  the
continuous time  equilibrium   quantum Markov  process for the Hamiltonian $A$. It corresponds to the quantum thermodynamical equilibrium for the action of the Hamiltonian $A$.
\end{abstract}



A short version of this work was published in  Open Systems and Information Dynamics, Volume 30, Issue 04  2350018 (2023)


\section{Introduction}

We are interested in continuous time  stationary quantum Markov process which corresponds to equilibrium for a quantum bath (interacting with a quantum system)  under the action of a certain given Hamiltonian.  Therefore, our results concern continuous-time quantum channels.

In \cite{CaMa} the authors present a detailed study of a nice version of the detailed balance condition for a continuous-time quantum Markov semigroup on $M_n(\mathbb{C})$, $n \in \mathbb{N}$ (see also \cite{JPW}, \cite{FR} and \cite{Pillet} for related results). In Theorem 4.2 in \cite{CaMa} it is explained that the detailed balance condition for a classical continuous-time  Markov Chain, with values on a finite state space, corresponds to the commutative part of the dynamical evolution of the continuous-time quantum Markov semigroup. Results of \cite{CaMa} and the detailed balance condition are used here in an essential way.

On page  75 in Section 9 in \cite{Kac}, and also on page 114 in Section 5 in \cite{Str},  for a  classical continuous-time  Markov Chain - with  infinitesimal generator $L$ - satisfying the detailed balance condition, a deviation function  (which is a form of entropy) is introduced and a variational principle (in some sense a  form of maximizing pressure) is considered  (see expression 9.18 in page 76 in \cite{Kac}). We would like to extend the results obtained  for the classical commutative setting to  the non-commutative setting  of quantum Markov semigroups satisfying the  detailed balance condition as described in \cite{CaMa}.

We will present a natural concept of entropy for a class of density matrices  (see Section \ref{enter}).
We point out that the dynamics of the flow in the set of matrices is encapsulated on the infinitesimal generator $\mathcal{L}$  and the entropy we consider here  is at the level of  this linear operator. In this sense, this concept of  entropy has no {\it direct} dynamical content. Our setting is the quantum channel version of the classical ones considered in \cite{Kac} and \cite{Str}.

After introducing entropy we will study a version of the variational principle of pressure  and its relation to an eigenvalue problem for a certain type of  transfer operator (see Section \ref{PPre} and expression \eqref{mamai6} in Section \ref{trans}). In classical Thermodynamic Formalism, the Ruelle operator plays this role. The Ruelle Theorem describes a relation of equilibrium states with a corresponding eigenvalue problem (see \cite{PP}).  The Ruelle operator is an infinite-dimensional version of the Perron-Frobenius operator. The transfer operator we consider here is not exactly an extension of the concept of Ruelle operator.
A density matrix  maximizing pressure will be called an equilibrium density matrix. We will provide examples in Section \ref{PPre}.

Our results are in some sense the quantum analogous of the reasoning delineated in  \cite{BEL}, \cite{LNT} and \cite{KLMN}, which considered the dynamics of  continuous-time dynamics (a flow) in the Skorohod space.

Our definition of entropy is also different from the ones in \cite{BKL1} and \cite{BKL2} which considered quantum channels in discrete-time dynamics.

Taking into account  the   concept and the notation described in section 5 in \cite{CaMa} we will denote by $ \mathcal{L}_0 $ (the Laplacian)  the generator of the heat semigroup. We will choose a special $ \mathcal{L}_0 $ (see Definition \ref{dde}) which will play the role of the {\it a priori} Laplacian.  Our special choice of $\mathcal{L}_0$ is the analogous of taking the normalized counting probability as  the {\it a priori}   probability in the classical definition of Kolmogorov-Shannon entropy (see discussion in \cite{LMMS}).

From  this $\mathcal{L}_0$ (which is fixed from now on) we will be able to define the detailed balanced condition (as described in \cite{CaMa}) and the  Laplacian-entropy.

\begin{definition} \label{defi} Given a density operator $\rho$ define the {\it Laplacian-entropy (entropy for short) by}
\begin{equation} \label{ent} h(\rho) =\,  \text{Tr } \,   [\,\rho^{1/2} \mathcal{L}_0^\dag (\rho^{1/2})\,],
\end{equation}
where we set $ \mathcal{L}_0$ by Definition \ref{dde}.
\end{definition}

 Our  definitions of entropy and pressure are quite natural. They are the non commutative extension of the concepts  considered in  the  classical setting of continuous-time Markov chains as described by M. Kac  in expressions (9.16) and (9.18) in Section 9  in \cite{Kac} and by D. W. Strook in Section 5 in \cite{Str}.

Expression (5.12) in \cite{Str} defines the so-called {\it rate function} $I$ in the setting of classical continuous-time  Markov Processes taking values on the compact metric space $E$. If $L$ is the infinitesimal generator, then  (5.12) means

$$\quad I(\nu)= - \inf_{u>0, \, u \in C(E)} \int \frac{Lu}{u} d\nu,$$
where $C(E)$ is the set of real continuous functions defined on $E$.

Later, for reversible processes, the above  formula simplifies to expression (5.18) in \cite{Str}, which claims
$$ I(\nu) = - \int \phi^{1/2} L\phi^{1/2} d\mu, \quad \phi=\frac{d\nu}{d\mu}.$$

In \cite{Str} it is used the term symmetric operator but  in other contexts, this would correspond to conditions like reversibility or the detailed balanced condition.

Under the  detailed balanced condition, in the quantum channel context, one should replace the role of $L$ by the generator of a QMS, which is usually  denoted by $\mathcal{L}$ (the Lindbladian). Probabilities are replaced by densities $\rho$ (states). In this case,  (5.12) in \cite{Str} corresponds here to
$$ I(\rho) = - \inf_{U>0} \tr(\rho\, U^{-1}\,\mathcal{L}_0\,U), $$
\noindent where the  infimum is taken over the positive matrices $U\in M_n$.

In \cite{Str} (see also \cite{Kac}) the variational principle is taken as
$$\lambda(V) = \sup_{\text{prob}\,\,\nu} (\int V d\nu - I(\nu)),$$
where $\lambda(V)$ is the main eigenvalue of a certain operator.

Denote
$D_{n} = \{\rho \ge 0:\tr(\rho)=1\}$.

In section \ref{PPre}  we consider an analogous problem:
given an Hermitian operator $A: \C^n \to \C^n$, we consider the variational problem:

\begin{equation} \label{PP21} P_{A}= \sup_{\,\rho\in D_{n}} \{h(\rho) + \tr(A\rho)\,\}.
\end{equation}

A matrix $\rho_A$  maximizing  $P_{A}$  will  be called an equilibrium density the operator for $A$. We call $P(A)$ the Laplacian-pressure (pressure for short) of the Hermitian operator $A: \C^n \to \C^n$.

A connection of $P_A$ of expression \eqref{PP21} with the  eigenvalue of a certain linear  operator is described in expression \eqref{mamai6} in Section \ref{trans} (see also \eqref{mamai5}). In this way we get  all elements for establishing a continuous time quantum channel version of the classical Ruelle operator (see \cite{PP}, \cite{LMMS}, \cite{BEL}, \cite{LNT}).

Our definition of entropy has a difference of sign when compared with the setting of \cite{Str}, so we wonder if there exists a connection between $h(\rho)$ and $-I(\rho)$.
In section \ref{Kaka} we will show this connection in the special case of the heat-semigroup with the {\it a priori}  generator $\mathcal{L}_0$ defined on section \ref{enter}. We will show that:

\begin{theorem}\label{a1} Given the density matrix $\rho$, then
$$ h(\rho) =  \inf_{A>0} \tr(\rho \,A^{-1}\, \mathcal{L}_0(A)).$$
\end{theorem}

Following \cite{CaMa}, we will present  in Section \ref{cla} the classical Markov Chain associated with a continuous-time quantum channel and we will provide examples.

The present work is part of the PhD. thesis of J. Brasil and J. Knorst in 
Graduate Program in Math. in UFRGS.

\section{An outline of the main prerequisites} \label{trr}

Given a  linear operator $A: \mathbb{C} \to \mathbb{C}$, its dual (with respect to the canonical inner product), is denoted by $A^*: \mathbb{C} \to \mathbb{C}$.

Denote by
$M_n(\mathbb{C})=M_n$ the set of  $n$ by $n$ complex matrices with the GNS inner product $<A,B> = \tr \, (A^* \, B)$. Given a linear operator $T: M_n \to M_n$, its dual with respect to this inner product is denoted by $T^\dag: M_n \to M_n.$ That is, for all matrices $A,B$ we get
$$ <T(A), B> = <A, T^\dag (B)>.$$

We denote by $\mathfrak{1}$ the diagonal matrix with entries $\frac{1}{n}$. Then, $\mathfrak{1}$ is a density matrix and also the unity of the $C^*$-algebra $M_n(\mathbb{C})$. We denote by $\mathfrak{G}_{+}$ the set of invertible density matrices (operators) $\rho: M_n \to M_n$.

We will consider  {\bf continuous time quantum  semigroups} (QS) the ones given by
$\mathcal{P}_t= e^{t\, \mathcal{L}}$, $t \geq 0$, where $\mathcal{L}:M_n(\mathbb{C}) \to M_n(\mathbb{C})$ is the infinitesimal generator (see Definition 5.5.1 and also Section 9.3.2 in \cite{Chang}).
The linear operator $\mathcal{L}$ should satisfy the conditional complete
positivity property (see section 5 in \cite{Chang} or section 6 in \cite{Wolf}). We assume that $\mathcal{L}(A^*) =(\mathcal{L}(A))^*$, for all $A \in M_n(\mathbb{C})$. Given a selfadjoint matrix $A\in M_n(\mathbb{C})$, the dynamical evolution
$t \to e^{t \mathcal{L}} (A)$ is called the Heisenberg dynamical evolution.
Given a density matrix $\rho\in M_n(\mathbb{C})$, the dynamical evolution
$t \to e^{t \mathcal{L}^\dag} (\rho)$ is called the  Schrödinger dynamical evolution.

If we assume that $\mathcal{L}(I)=0$, we will call $\mathcal{P}_t= e^{t\, \mathcal{L}}$, $t \geq 0$, the continuous time {\bf quantum Markov  semigroup} (QMS)  associated to $\mathcal{L}$ (see Definition 5.5.2  and also section  7 in \cite{Chang}). It is known that in this case   $ e^{t\, \mathcal{L}}(\mathfrak{1})= \mathfrak{1}$, for all $t\geq 0.$ Continuous time  quantum Markov semigroups provide a convenient mathematical description of the irreversible
dynamics of an open quantum system.

If $\rho=\rho_\mathcal{L}$ is such that $\mathcal{L}^\dag (\rho)=0$, then for all $t \geq 0$, we get $ e^{t\, \mathcal{L}^\dag}(\rho)= \rho$ and we say that $\rho$ is the {\bf stationary density matrix for the  continuous time  quantum Markov semigroup} with infinitesimal generator $\mathcal{L}$. Section 9.4 in \cite{Chang} presents a discussion on the uniquenes of the stationary matrix $\rho$.

We call  $X_t$, $t\geq 0$,  the {\bf continuous time  quantum Markov  process} (QMP) associated to the infinitesimal generator $\mathcal{L}$,
the process associated to the pair $(e^{t\, \mathcal{L}}$, $\rho_\mathcal{L})$, $t\geq 0$. We can ask questions about ergodicity for such process (see Section 11  in \cite{Chang}).

Given the (QMP) associated to $\mathcal{L}$ and the stationary density operator $\rho$,
take an observable (a self-adjoint matrix) $A\in M_n$. Then, we get  that
$$ t \to \text{Tr} (\rho \, e^{t \mathcal{L}} (A))$$
describes the time evolution of the expected value of the observable $A$.

We say that $\mathcal{L}$ is irreducible if for every non-zero matrix $A \geq 0$, and every strictly positive $t>0$, we have $e^{t\, \mathcal{L}}(A)> 0.$ We will also assume that $\mathcal{L}$ is irreducible (see Sections 10 and  11 in \cite{Chang}).

Given $\sigma \in \mathfrak{G}_{+}$, consider the inner product $<\,,\,>_\sigma$ in the set of matrices in $M_n$ given by
$<A,B>_\sigma = \tr \, (A^* B \sigma)=<A, B \sigma>.$

\begin{definition} Given $\sigma \in \mathfrak{G}_{+}$, we say that the QMS $e^{t \mathcal{L}}, t \geq 0$, satisfies the $\sigma$-detailed balance condition if $\mathcal{L}$ is symmetric with respect to $<\,,\,>_\sigma$. That is, for all matrices $A,B\in  M_n$ we get
$$ < \mathcal{L} (A), B>_\sigma =<  A,  \mathcal{L}(B)>_\sigma. $$

\end{definition}

Given $\sigma \in \mathfrak{G}_{+}$, if the QMS $e^{t \mathcal{L}}, t \geq 0$, satisfies the $\sigma$-detailed balance condition, then, $\sigma$ is stationary for the evolution of the semigroup $e^{t \mathcal{L}^\dag}, t \geq 0$ (see Lemma \ref{iin}).

The explicit form of the infinitesimal generator of a  continuous-time quantum Markov  semigroups  satisfying the detailed balance condition is described by
expression (3.4) in \cite{CaMa} (see our expression \eqref{gre}).

\begin{definition}
\label{dde}
We denote $\mathcal{L}_0$ the infinitesimal generator  satisfying d.b.c.   where we take $\sigma=\mathfrak{1}$. $\mathcal{L}_0$  will be called the Laplacian (see Section 3 in \cite{CaMa}).
\end{definition}

The semigroup   $\mathcal{P}_t= e^{t\, \mathcal{L}_0}$, $t \geq 0$, describes the unperturbed continuous time quantum channel.

Given the {\it a priori} Laplacian operator $\mathcal{L}_0:M_n(\mathbb{C}) \to M_n(\mathbb{C})$,
we will present a natural concept of entropy for a class of density matrices $\rho$ on $M_n(\mathbb{C})$ (see Definition \ref{defient}).

Given a Hermitian operator  $A:\mathbb{C}^n\to \mathbb{C}^n$ (which plays the role of minus the  Hamiltonian), we will consider in Section \ref{PPre} a variational principle of pressure for $A,$ which is given by Definition \ref{defipre}.

A density matrix $\rho_A$ maximizing pressure will be called an equilibrium density matrix for $A$. This matrix (in fact $\rho_A^{1/2}$) will satisfy an eigenvalue property  for a certain linear operator $\mathfrak{L}_A$  to be described in Section \ref{trans}. From $\rho_A$ we will derive a  new  infinitesimal generator $\mathcal{L}_A$. Finally,   the continuous-time quantum Markov  Process $X_t$, $t \geq 0$,    associated to  $\mathcal{P}_t= e^{t\, \mathcal{L}_A}$, $t \geq 0$, and $\rho_A$, will be called  the
continuous-time  equilibrium   quantum Markov  semigroup for the Hamiltonian $A$. This new process describes a continuous-time quantum channel after the perturbation by the selfadjoint operator $A$.

      \section{The heat semigroup and entropy of density operators}  \label{enter}

      In this section the inner product in $M_n$ is $<A,B>=$ Tr $(A^* \,B)$.

      Denote by $e_j$, $j=1,...,n$, the canonical base in $\mathbb{C}^n,$
      and by $$V_{i,j} = \outerp{e_i}{e_j}:\mathbb{C}^n \to \mathbb{C}^n,$$
      where  $i, j=1,...,n$. Note that $V_{i,j}^* = \outerp{e_j}{e_i}$.

We denote by $V_{i,j}$ the matrix which is zero in all entries, up to the entry $i,j$, where it has the value $1$.

We denote by $\mathfrak{1}$ the operator identity $I$ times $\frac{1}{n}.$ The matrix $\mathfrak{1}$ describes an invertible density operator.

$\mathcal{L}_0$ denotes  the infinitesimal generator  satisfying d.b.c.   for $\sigma=\mathfrak{1}$.

One can show that $V_{i,j},$ $i,j=1,...,n$,
is an orthonormal basis for  $\mathcal{L}_0$ associated to the eigenvalue $0$.

      Following section 5 in \cite{CaMa} we call $ \mathcal{L}_0 $ (the Laplacian)  the generator of the heat semigroup (the Laplacian)
    \begin{equation} \label{grer}
        A \to \mathcal{L}_0 (A) =
         \sum_{i,j=1}^{n}  \,( V_{i,j}^*\,[A,V_{i,j}] +   \,[V_{i,j},A]\,  V_{i,j}^*)
    \end{equation}

This operator is negative-semi definite.
(see page 1827 in \cite{CaMa} and also \eqref{gre} and \eqref{gre1} of next section. Note that $ \mathcal{L}_0 (I)=0.$

One can show that $ \mathcal{L}_0^\dag =  \mathcal{L}_0 $ and Tr $(\mathfrak{1}\, \mathcal{L}_0 (A))=0$, for all $A \in M_n$.

Note  that
\begin{equation} \label{grert}\mathcal{L}_0^\dag (\rho) =     \sum_{i,j=1}^{n}   \,([ V_{i,j}\, \rho, V_{i,j}^*] +   \,[V_{i,j},\rho \, V_{i,j}^*]\, ).
\end{equation}

$\sigma=\mathfrak{1}$ is  invariant for the flow $e^{ t \mathcal{L}_0^{\dagger}}$.

\begin{definition} \label{defient} Given a density operator $\rho$ define the {\it Laplacian-entropy}
\begin{equation} \label{ent} h(\rho) =\,  \tr \, [\,\rho^{1/2} \mathcal{L}_0^\dag (\rho^{1/2})\,]
\end{equation}
\end{definition}

This definition is consistent with expression (5.18) on page 113 in \cite{Str}.

Our main result in this section is the explicit expression for entropy to be described  by Proposition \ref{h_eigen_dependence}.
\medskip

First, we want to show the following Lemma:

\begin{lemma}
$h( \mathfrak{1})=0.$
\end{lemma}

\begin{proof}
From \eqref{grert}
$$\mathcal{L}_0^\dag (\rho^{1/2}) =     \sum_{i,j=1}^{n}   \,([ V_{i,j}\, \rho^{1/2}, V_{i,j}^*] +   \,[V_{i,j},\rho^{1/2} \, V_{i,j}^*]\, ).$$

Then,
$$\mathcal{L}_0^\dag (\mathfrak{1}^{1/2}) =  \sum_{i,j=1}^{n} \,([ V_{i,j}\, \mathfrak{1}^{1/2}, V_{i,j}^*] +   \,[V_{i,j},\mathfrak{1}^{1/2} \, V_{i,j}^*]\, )= \mathfrak{1}^{1/2} \, 2 \sum_{i,j=1}^{n} \,([ V_{i,j}, V_{i,j}^*].$$

Note that
$$ \,[ V_{i,j}, V_{i,j}^*] = $$
$$   \,  \outerp{e_i }{ e_i}   -    \outerp{e_j }{ e_j}\, .$$
For each pair $(i,j)$ there is a correspondent $(j,i).$ From this follows that
$$\mathcal{L}_0^\dag (\mathfrak{1}^{1/2}) =    \mathfrak{1}^{1/2} \, 2 \sum_{i,j=1}^{n}      \outerp{e_i }{ e_i}   -   \outerp{e_j}{e_j} =0.$$

Then, $h( \mathfrak{1})=0.$

\end{proof}
\bigskip

Consider now a general density operator $\rho\in \mathfrak{S}_{+}.$  We want to estimate $h(\rho).$

Denote $\rho^{1/2} =  \sum_{r,s=1}^{n}  c_{r s}  \,  \outerp{e_r}{e_s}   $.


For $i,j$ fixed
$$  [V_{i,j}\, \rho^{1/2}, V_{i,j}^*] = \big[V_{i,j}\,  \big(\,\sum_{r,s=1}^{n}  c_{r s}  \,  \outerp{e_r}{e_s} \,\big), V_{i,j}^*\big] =  $$
$$ V_{i,j}\,  \big(\,\sum_{r,s=1}^{n}  c_{r s}  \,  \outerp{e_r}{e_s} \,\big)\,V_{i,j}^* - \,V_{i,j}^*\, V_{i,j}\,  \big(\sum_{r,s=1}^{n}  c_{r s}  \,  \outerp{e_r}{e_s} \,\big)=  $$
$$ V_{i,j}\,  \big(\sum_{r,s=1}^{n}  c_{r s}  \,  \outerp{e_r}{e_s} \,\big)\,V_{i,j}^* - \outerp{e_j}{e_j} \,  \big( \sum_{r,s=1}^{n}  c_{r s}  \,  \outerp{e_r}{e_s} \,\big)=  $$
$$ V_{i,j}\,  \big(\sum_{r,s=1}^{n}  c_{r s}  \,  \outerp{e_r}{e_s} \,\big)\, \outerp{e_j}{e_i} - \,\sum_{s=1}^{n}  c_{j s}  \,  \outerp{e_j}{e_s} \,=  $$


$$ \outerp{e_i}{e_j} \,\sum_{r=1}^{n}  c_{r j}  \,  \outerp{e_r}{e_i}   - \,\sum_{s=1}^{n}  c_{j s}  \,  \outerp{e_j}{e_s} \,=  $$
$$ c_{j\, j}\,\outerp{e_i}{e_i}  - \,\sum_{s=1}^{n}  c_{j s}  \,  \outerp{e_j}{e_s} \,.  $$

On the other hand, for $i,j$ fixed
$$  [V_{i,j}, \, \rho^{1/2}\, V_{i,j}^*] = \big[V_{i,j}\, ,\, \big(\sum_{r,s=1}^{n}  c_{r s}  \,  \outerp{e_r}{e_s} \,\big)\, V_{i,j}^*\big] =  $$
$$ V_{i,j}\,  \big(\sum_{r,s=1}^{n}  c_{r s}  \,  \outerp{e_r}{e_s} \, \big)\,V_{i,j}^* - \,  \big(\sum_{r,s=1}^{n}  c_{r s}  \,  \outerp{e_r}{e_s} \,\big)\,V_{i,j}^*\, V_{i j}\,=  $$
$$ c_{j\, j}\,\outerp{e_i}{e_i}  -   \big( \sum_{r,s=1}^{n}  c_{r s}  \,  \outerp{e_r}{e_s} \,\big)\outerp{e_j}{e_j} \,=  $$
$$  c_{j\, j}\,\outerp{e_i}{e_i}  - \,\sum_{r=1}^{n}  c_{r j}  \,  \outerp{e_r}{e_j} \,.  $$

Finally,
$$\mathcal{L}_0^\dag (\rho^{1/2}) =     \sum_{i,j=1}^{n}   \,\big([ V_{i,j}\, \rho^{1/2}, V_{i,j}^*] +   \,[V_{i,j},\rho^{1/2} \, V_{i,j}^*]\, \big)=$$
$$  \sum_{i,j=1}^{n} \Big(\, c_{j\, j}\,\outerp{e_i}{e_i}  - \,\sum_{s=1}^{n}  c_{j s}  \,  \outerp{e_j}{e_s} \, +  $$
$$ + \, c_{j\, j}\,\outerp{e_i}{e_i}  - \,\sum_{r=1}^{n}  c_{r j}  \,  \outerp{e_r}{e_j} \Big) = $$
$$  2 \,  \sum_{j=1}^{n} \,c_{j\, j}\, I - \sum_{i,j=1}^{n} \,\,\Big( \,    \sum_{s=1}^{n}  c_{j s}  \,  \outerp{e_j}{e_s}  +  \,\sum_{r=1}^{n}  c_{r j}  \,  \outerp{e_r}{e_j}    \, \Big)  =$$
$$ 2 \,  \sum_{j=1}^{n} \,c_{j\, j}\, I - n \sum_{j=1}^{n} \,\,\Big( \,    \sum_{s=1}^{n}  c_{j s}  \,  \outerp{e_j}{e_s}  +  \,\sum_{r=1}^{n}  c_{r j}  \,  \outerp{e_r}{e_j}    \, \Big).$$
Then,
$$\rho^{1/2} \mathcal{L}_0^\dag (\rho^{1/2})=$$
$$ \sum_{u,v=1}^{n}  c_{u v}  \,  \outerp{e_u}{e_v}  \,\,\, \,2 \,  \sum_{j=1}^{n} \,c_{j\, j}\, I\,-  $$
$$n\, \sum_{u,v=1}^{n}  c_{u v}  \,  \outerp{e_u}{e_v}  \,\,\Big(\,\sum_{j=1}^{n}  \,\,\,  \,\sum_{s=1}^{n}  c_{j s}  \,  \outerp{e_j}{e_s}  + \,\sum_{j=1}^{n}  \,\sum_{r=1}^{n}  c_{r j}  \,  \outerp{e_r}{e_j}   \,\Big) =$$
$$\sum_{u,v=1}^{n}  c_{u v}  \,  \outerp{e_u}{e_v}  \,\,\, \,2 \,  \sum_{j=1}^{n} \,c_{j\, j}\, I\,- $$
$$n\, \sum_{u,v=1}^{n}  c_{u v}  \, \ket{e_u}\Big( \sum_{s=1}^{n}\, c_{vs}\,\bra{e_s} \Big)   +  n\, \sum_{u,v=1}^{n}  c_{u v} \, \ket{e_u} \Big( \sum_{j=1}^{n}  c_{v j}  \,  \bra{e_j} \Big)  \,\, =$$
$$\sum_{u,v=1}^{n}  c_{u v}  \,  \outerp{e_u}{e_v}  \,\,\, \,2 \,  \sum_{j=1}^{n} \,c_{j\, j}\, I\,- 2\, n\,\, \sum_{u,v=1}^{n}  c_{u v}  \,\ket{e_u}\,\Big(\sum_{s=1}^{n}\, c_{v s}\, \bra{e_s} \Big)    =$$
$$ 2\,\sum_{j=1}^{n} \,c_{j\, j} \,\sum_{u,v=1}^{n}  c_{u v}  \,  \outerp{e_u}{e_v}  \,\,\,  - 2\, n\,\, \sum_{u,v=1}^{n}\,\sum_{s=1}^{n}c_{u v} \,c_{v s} \, \outerp{e_u}{e_s}.
$$

The diagonal of $\rho$ is $\rho_{11},\rho_{22},...,\rho_{nn}$ and $\sum_{j=1}^{n}     \,c_{j\, j}= \sum_{j=1}^{n}  \sqrt{\lambda_j}$, where $\lambda_j$, $j=1,...,d$, are the eigenvalues of $\rho$. Note that $\sup (\sum_{j=1}^{n}  \sqrt{\lambda_j})^2 = n$ and $\inf (\sum_{j=1}^{n}  \sqrt{\lambda_j})^2 =1$.

Therefore, for fixed $n$ and a given $\rho$
$$h(\rho)=\tr \, [\rho^{1/2} \mathcal{L}_0^\dag (\rho^{1/2})]=  2\,\sum_{j=1}^{n} \,c_{j\, j}\sum_{u=1}^{n} \,c_{u\, u}  -2\, n\,\sum_{v=1}^{n}\, \sum_{s=1}^{n} c_{s v} c_{v s}=$$
$$
2\,\Big(\sum_{j=1}^{n}     \,c_{j\, j}\Big)^2 -2\, n\,\sum_{s=1}^{n} \rho_{s s} =  2\,\Big(\sum_{j=1}^{n}  \sqrt{\lambda_j}\Big)^2     - 2 n \leq 0. $$

Note $h(\rho)$ can be very negative if $n$ is large.
We get  the following proposition by looking at this last inequality:

\begin{proposition}\label{h_eigen_dependence}
  The entropy $h$ depends only on $\{\lambda_{i}\}$ the eigenvalues of $\rho$ and
  \[h(\rho) = 2\tr \,(\rho^{1/2})^{2} - 2n = 2\left(\sum_{j=1}^{n}\sqrt{\lambda_{i}}\right)^{2} - 2n.\]
\end{proposition}

Note that as $\sum_{j=1}^{n} \lambda_j=1$, the maximal value of $h(\rho) $ is zero, and this happens  when all eigenvalues $\lambda_j=1$  are equal to $\frac{1}{n}$. The maximal value of entropy is attained by the density matrix $\mathfrak{1}$.

\bigskip

\begin{remark} \label{rere}
For fixed $A$ we denote $\partial_{i,j} (A) = [V_{i,j}, A]$ and $\partial_{i,j}^\dag (A) = [V_{i,j}^*, A].$

$\partial_{i,j}$ is a version of the momentum operator $\frac{1}{i} \frac{\partial}{\partial x}$ acting on the set $L^2$ of functions for the Lebesgue probability on the circle.  Indeed, denote by $D$ the operator $g \to D(g) =\frac{1}{i} g'$.
For fixed $a:[0,1) \to \mathbb{R},$  take the multiplication operator  $g \,\to  a \, g$ acting on functions $g$. Then, the operator

$$g \to D( a g) - a D(g) = \frac {1}{i}  a' \,g,$$
describes multiplication by  $\frac {1}{i}  a' = \frac{1}{i} \frac{\partial \,a}{\partial x}$.


We point out that   $\sum_{i,j} \partial_{i,j} \partial_{i,j}^\dag $ corresponds to second derivative (Laplacian). On the other hand    $\sum_{i,j} \partial_{i,j} \partial_{i,j} $ corresponds to minus second derivative (minus Laplacian).

\end{remark}

\section{The general setting for detailed balanced condition}

Before we begin the study of the quantum case we will state results for the detailed balanced condition when the
continuous   time Markov Chain takes values on $\{1,2..,k\}$. Denote by $\mathfrak{s}= (\mathfrak{s}_1,..., \mathfrak{s}_k)$ the initial invariant probability for the line sum zero matrix
$W =(W_{i,j})_{i,j=1,...,k}$.

The detailed balance condition for $W$ is: for all $i,j =1,...,k$
$$ \mathfrak{s}_i W_{i,j} = \mathfrak{s}_j W_{j,i}.$$

Consider the inner product in $\mathbb{R}^k$

$$<x,y>_\mathfrak{s} = \sum_{j=1}^k \mathfrak{s}_j x_j y_j. $$

It is easy to see that $W$ satisfies the detailed balance condition, if and only if, $W$ is self-adjoint for the inner product $<.,.>_\mathfrak{s}$.

The above is the classical (commutative) setting for presenting the  detailed balance condition. We are interested in presenting the non-commutative version of the concept.

\medskip

We will be interested here in the $C^*$-Algebra $\mathcal{A}=M_n$ of complex $n$ by $n$ matrices. The inner product in $M_n$ is $<A,B> = \tr \, [A^*\, B] $.
Following the notation of \cite{CaMa} the associated Hilbert space  will be denoted by $\mathfrak{h}_\mathcal{A}$.

 We will fix from now on an element $\sigma:\mathbb{C}^n \to \mathbb{C}^n$ in $\mathfrak{S}_{+}$. A hypothesis that can be helpful for ergodic properties  is to assume that all eigenvalues of $\sigma$ are simple.

 \medskip

 Now, we present some preliminary definitions and properties taken from \cite{CaMa}.
 \smallskip

Once we fix the Hamiltonian $H$ we fix the density state $\sigma$ via $\sigma= e^{-h} $ (in some sense we are considering a ``normalized'' Hamiltonian). In fact, $\sigma= \frac{e^{-H}}{Tr (e^{-H})}$ and $h=H + \log Tr (e^{-H}).$

The linear transformation, $\Delta_\sigma: M_n \to M_n$, is given by $A \to \Delta_\sigma(A) = \sigma \, A \, \sigma^{-1}$. Note that
$\Delta_\sigma (\sigma) =\sigma$ and $\Delta_\sigma (\mathfrak{1})=\mathfrak{1}.$

Assume each
eigenvalue of $\sigma$ is simple.

$\Delta_\sigma^{-1}: M_n \to M_n$, is given by $B \to \Delta_\sigma(B) = \sigma^{-1} \, B \, \sigma.$

Note that $ (\Delta_\sigma (A))^* = \Delta_\sigma^{-1} (A^*).$

$\mathcal{K}: M_n \to M_n$ is positive preserving, if  $ \mathcal{K} (A)\geq 0$, in the case $A\geq 0.$

$\Delta_\sigma$ is positive but not positive preserving.

We say that
$\mathcal{K}: M_n \to M_n$ is self-adjointness preserving, if  $ (\mathcal{K} (A))^*=  \mathcal{K} (A^*)$.

Remember that by definition $\mathcal{K}^\dag: M_n \to M_n$  is the one such that for all $A,B$
$$Tr  [ A^* \mathcal{K}(B)   ] =  <A, \mathcal{K}(B) >=< \mathcal{K}^\dag(A) ,B>= Tr [(\mathcal{K}^\dag(A))^*\, B].$$

$\mathcal{K}: M_n \to M_n$ is self-adjoint if $\mathcal{K}=\mathcal{K}^\dag.$

$\Delta_\sigma$ is self-adjoint.

Note that if $\Delta_\sigma (E) = e^{- w} E$, then $\Delta_\sigma (E^*) = e^{ w} E^*.$

Assume that $\eta_j\in \mathbb{C}^n$, $j=1,2,...,n$, is an orthonormal basis for $h =-  \log \sigma$,
\begin{equation} \label{tytya} h (\eta_j)= \lambda_j \eta_j,\,\,\, \forall j.
\end{equation}

Then,  $\eta_j\in \mathbb{C}^n$, $j=1,2,...,n$, is also an orthonormal basis for $\sigma$.

$h= - \log  \sigma$ plays the role of the Hamiltonian.

Then, we get that
$\eta_j\in \mathbb{C}^n$, $j=1,2,...,n$, is an orthonormal basis for $\sigma$,
\begin{equation} \label{tyty} \sigma (\eta_j)= e^{-\, \lambda_j} \eta_j,\,\,\, \forall j.
\end{equation}

 As  $\sigma:\mathbb{C}^n \to \mathbb{C}^n$ is in
$\mathfrak{S}_{+}$, we get that
\begin{equation} \label{tutu}\sum_{j=1}^n e^{-\, \lambda_j}=1.
\end{equation}

For $\alpha=(\alpha_1,\alpha_2) $, where $\alpha_1,\alpha_2 \in \{1,2,...,n\},$ denote
\begin{equation} \label{tyui }w_{\alpha_1,\alpha_2}=w_{\alpha} = \lambda_{\alpha_1} - \lambda_{\alpha_2}.
\end{equation}

If $\sigma= \mathfrak{1}$, then, all $w_{\alpha_1,\alpha_2}=0.$

For each pair $\alpha \in  \{1,2,...,n\}^2$, denote
$$F_\alpha = \outerp{\eta_{\alpha_1} }{\eta_{\alpha_2}}.$$
Note that
\begin{equation} \label{ty}F_{  (\alpha_1,\alpha_2)}^* = F_{  (\alpha_2,\alpha_1)}.
\end{equation}
Moreover,

\begin{equation} \label{tuo} \Delta_\sigma (F_\alpha)=    e^{-\lambda_{\alpha_1}+\lambda_{\alpha_2}  }\,F_\alpha.
\end{equation}

Indeed,
$$            \Delta_\sigma (F_\alpha) = \sigma\, \outerp{\eta_{\alpha_1} }{\eta_{\alpha_2}}\, \sigma^{-1}= e^{-\lambda_{\alpha_1} } \, \outerp{\eta_{\alpha_1} }{ \eta_{\alpha_2}} \, \sigma^{-1} =$$
$$  e^{-\lambda_{\alpha_1}+\lambda_{\alpha_2}  } \, \outerp{\eta_{\alpha_1} }{ \eta_{\alpha_2}}=e^{-\lambda_{\alpha_1}+\lambda_{\alpha_2}  }\,F_\alpha.$$

\smallskip

This shows that:
\begin{lemma} \label{lelei}
The operators $F_\alpha = \outerp{\eta_{\alpha_1} }{\eta_{\alpha_2}}$, $\alpha_1,\alpha_2 \in \{1,2,...,n\},$  describe a natural orthonormal basis of $\Delta_\sigma.$   The corresponding eigenvalues are $ e^{-\lambda_{\alpha_1}+\lambda_{\alpha_2}  }$.
\end{lemma}

\smallskip

Note that if $\tau$ represents the normalized trace and $\alpha=(\alpha_{1}, \alpha_{2})$ is such that $\alpha_{1}\ne\alpha_{2}$, then,
$$\tau (F_\alpha)=0$$
and, for $\alpha,\tilde{\alpha}\in  \{1,2,...,n\}^2$
$$\tau (F_\alpha^*\, F_{\tilde{\alpha}})= \delta_{\alpha,\tilde{\alpha}} := \delta_{\alpha_{1},\tilde{\alpha}_{1}}\delta_{\alpha_{2},\tilde{\alpha}_{2}}.
$$

Now we denote the different $F_\alpha$, $\alpha \in  \{1,2,...,n\}^2$, by $V_{k}$, $k=1,...,n^2$ (in order to use the same notation as  in \cite{CaMa}).  In this identification, we also denote for
each $k=1,...,n^2$, the value
$w_k =\lambda_{\alpha_1}-\lambda_{\alpha_2}  $, for the corresponding $\alpha=(\lambda_{\alpha_2},\lambda_{\alpha_1}).$

Then, the family $V_1,...,V_{n^2}$ represent the different eigenmatrices (an orthonormal basis) for  $\Delta_\sigma$ associated to the eigenvalues $e^{- w_1}, ...,e^{- w_{n^2}}$,
where $w_k\in \mathbb{R}$, $k=1,...,n^2$. $\mathfrak{1}$ and $\sigma$ are eigenmatrices associated to the eigenvalue $1$. The matrices $V_{k}$ do  not have to be
self-adjoint, but from \eqref{ty} we get
$$\{ V_1,...,V_{n^2} \} = \{V_1^*,...,V_{n^2}^* \}.$$

Therefore, if $w_k$ is in the above list, there exists a $j$ such that  $w_j= - w_k$.

Given the Hamiltonian $h= - \log  \sigma$,
the modular automorphism $\alpha_t: M_n \to M_n$, $t \geq 0$, is defined  by
$$ A \to \alpha_t (A) =e^{i t h} A \, e^{-i t h} \,\Leftarrow\,\text{\,\,\,The Heisenberg point of view}.$$

A Quantum Markov Semigroup (QMS) is a continuous one-parameter semigroup of linear transformations $\mathcal{P}_t : M_n \to M_n$, $t \geq 0$, such that for each $t \geq 0$,  $\mathcal{P}_t$ is completely positive and $\mathcal{P}_t (\mathfrak{1})=\mathfrak{1}$.

It is natural to focus on quantum Markov semigroups
that commute with the modular operator $\Delta_\sigma$ associated to their invariant states $\sigma$.

Consider a QMS $\mathcal{P}_t : M_n \to M_n$, $t \geq 0,$ of the form
$$\mathcal{P}_t = e^{t \mathcal{L}},$$
for some linear operator $\mathcal{L} : M_n \to M_n$.

The operator $\mathcal{L}$ acts on observables (self-adjoint matrices). Note that  $\mathcal{L}(\mathfrak{1})=0$. The dual operator  $\mathcal{L}^\dag$ acts on density matrices.

A state $\sigma$ is invariant if Tr $[\sigma \, \mathcal{L}(A)]=0$, for all $A \in M_n$.

In terms of the possible inner products described on Defnition 2.2 in \cite{CaMa} we will choose $s=1$.

Remember that given $\sigma$ we  consider the inner product $<\,,\,>_\sigma \,=\, <\,,\,>_1$
in $M_n$, where
$$<A,B>_\sigma \, = \tr \, [\sigma A^* B] .$$

From \cite{CaMa} we get:

\begin{proposition} \label{defb} Given the density operator $\sigma$,  the the QMS $\mathcal{P}_t : M_n \to M_n$, $t \geq 0,$ of the form
$\mathcal{P}_t = e^{t \mathcal{L}},$ satisfies the $\sigma$- detailed balance condition, if and only if,
\begin{equation} \label{Paz} \mathcal{L} \circ \Delta_\sigma = \Delta_\sigma \circ \mathcal{L} .
\end{equation}
\end{proposition}

If $\mathcal{P}_t = e^{t \mathcal{L}},$ satisfies the $\sigma$-detailed balanced condition, then, for all $t\geq 0$,
 $$  \mathcal{P}_t \circ \Delta_\sigma = \Delta_\sigma \circ  \mathcal{P}_t  .$$

 Moreover, for any $t,t'$ and matrix $A$ we get that
 $$(\alpha_{t'} \circ \mathcal{P}_t) (A) = (\mathcal{P}_t \circ \alpha_{t'})(A).$$

It follows from \eqref{Paz} that
$V_1,...,V_{n^2}$ is an orthonormal basis for  $\mathcal{L}$ associated to the eigenvalues $e^{- v_1}, ...,e^{- v_{n^2}}$,
\noindent where $v_j\in \mathbb{R}$, $j=1,...,n^2$.
\medskip

\begin{lemma} \label{iin} $\sigma$ is a stationary density matrix for the semigroup  with infinitesimal generator $\mathcal{L}.$

\end{lemma}\begin{proof}

Note that
\begin{align*}
  [ V_{i,j}\, \sigma, V_{i,j}^*] &= \outerp{\eta_{i}}{\eta_{j}}\sigma\outerp{\eta_{j}}{\eta_{i}}
  - \outerp{\eta_{j}}{\eta_{i}}\outerp{\eta_{i}}{\eta_{j}}\sigma\\
  &= e^{-\lambda_{j}}\outerp{\eta_{i}}{\eta_{i}} - e^{\lambda_{j}}\outerp{\eta_{j}}{\eta_{j}}\\
  &= e^{-\lambda_{j}}\left(\outerp{\eta_{i}}{\eta_{i}} - \outerp{\eta_{j}}{\eta_{j}}\right)
\end{align*}
and
\begin{align*}
  [V_{i,j},\sigma \, V_{i,j}^*] &= \outerp{\eta_{i}}{\eta_{j}}\sigma\outerp{\eta_{j}}{\eta_{i}}
  - \sigma\outerp{\eta_{j}}{\eta_{i}}\outerp{\eta_{i}}{\eta_{j}}\\
  &= e^{-\lambda_{j}}\outerp{\eta_{i}}{\eta_{i}} - e^{\lambda_{j}}\outerp{\eta_{j}}{\eta_{j}}\\
  &= e^{-\lambda_{j}}\left(\outerp{\eta_{i}}{\eta_{i}} - \outerp{\eta_{j}}{\eta_{j}}\right).
\end{align*}

Using  expression \eqref{gre1} we get

\begin{align*} \label{lldagsigmasimpli}
  \Ll^{\dag}(\sigma) &=2\sum_{i,j}  e^{(\lambda_{j} - \lambda_{i})/2}e^{-\lambda_{j}}
                       \left(\outerp{\eta_{i}}{\eta_{i}} - \outerp{\eta_{j}}{\eta_{j}}\right)\\
  &= 2\left[\sum_{j}e^{-\lambda_{j}/2}\sum_{i}e^{-\lambda_{i}/2}\outerp{\eta_{i}}{\eta_{i}}
    - \sum_{i}e^{-\lambda_{i}/2}\sum_{j}e^{-\lambda_{j}/2}\outerp{\eta_{j}}{\eta_{j}}\right]\\
  &= 0.
\end{align*}

\end{proof}

\medskip

Remember that for each pair $i,j\in  \{1,2,...,n\}$,  we denote
$$V_{i,j} = \ket{\eta_{i}}\bra{ \eta_{j}},$$
where $\eta_{i}$, $i\in  \{1,2,...,n\}$, is the orthonormal basis of eigenvectors for $h =- \log \sigma$ associated to the eigenvalues $\lambda_j$ (according to \eqref{tytya}). As we mention before $w_{i,j} =\lambda_i -\lambda_j,$ $i,j\in  \{1,2,...,n\}$.

\begin{theorem} If $\mathcal{P}_t = e^{t \mathcal{L}},$ satisfies the $\sigma$-detailed balanced condition, then $\mathcal{L}$ is of the form
\begin{equation} \label{gre} A \to \mathcal{L}(A) = \sum_{i,j=1}^{ n}  e^{- w_{i,j}/2}\,( V_{i,j}^*\,[A,V_{i,j}] +   \,[V_{i,j}^*,A]\,  V_{i,j}),
\end{equation}

\noindent where $V_{i,j} = \ket{\eta_{i}}\bra{ \eta_{j}}$ and
$w_{i,j}= \lambda_i -\lambda_j,$ $i,j\in  \{1,2,...,n\}$, are real numbers such that \eqref{tyty}, \eqref{tutu} and \eqref{tyui } are true (which also means $ \Delta_\sigma (V_{i,j})=    e^{-\lambda_{i}+\lambda_{j}  }\,V_{i,j}$).

Note that given $\sigma\in \mathfrak{G}_{+}$,  the eigenvetors  $\ket{\eta_{i}}$ and eigenvalues $\lambda_j$, $j\in  \{1,2,...,n\}$, are determined.
Therefore, if $\mathcal{P}_t = e^{t \mathcal{L}}$ satisfies the $\sigma$-detailed balanced condition, then, $\mathcal{L}$ is uniquely determined.
\end{theorem}

\begin{remark} Conversely, given $\sigma$ in $\mathfrak{G}_{+}$ and $V_j$, $j=1,2,...,n^2$, such that,

1. $\Delta_\sigma V_j = e^{- w_j} V_j$,

2. $\{V_j, j=1,2,...,n^2\}=  \{V_j^*, j=1,2,...,n^2\}, $

then,
\begin{equation} \label{gre111} A \to \mathcal{L}(A) = \sum_{j=1}^{ n^2}  e^{- w_{j }/2}\,( V_{j}^*\,[A,V_{j}] +   \,[V_{j}^*,A]\,  V_{j}),
\end{equation}
is the infinitesimal generator of a QMS $e^{t \mathcal{L}}$, $t \geq 0$, which satifies de d.b.c. for the given $\sigma$. Therefore, $\sigma$ is stationary for
$e^{t \mathcal{L}^\dag}$, $t \geq 0$.

\end{remark}

The dual operator $\mathcal{L}^\dag$  satisfies
\begin{equation} \label{gre1} \rho \to \mathcal{L}^\dag(\rho) = \sum_{i,j=1}^{ n}  e^{- w_{i,j}/2}\,([ V_{i,j}\, \rho, V_{i,j}^*] +   \,[V_{i,j},\rho \, V_{i,j}^*]\, )
\end{equation}

\begin{remark} \label{uii}
If $\mathcal{P}_t = e^{t \mathcal{L}},$ satisfies  the detailed balanced condition  for the $\sigma=\mathfrak{1}$, then from
\eqref{Paz} we get that $V_{i,j}=\mathfrak{I}_{i,j},$ $i,j=1,...,n$. This is the case when $\mathcal{L}= \mathcal{L}_0.$
\end{remark}

\section{The Pressure problem} \label{PPre}

\begin{definition} \label{defipre}
Given an Hermitian operator $A: \C^n \to \C^n$,
$D_{n} = \{\rho \ge 0:\tr(\rho)=1\}$,
consider

\begin{equation} \label{PP1}P_{A}(\rho) = h(\rho) + \tr(A\rho)
\end{equation}

and

\begin{equation} \label{PP2} P_{A}= \sup_{\rho\in D_{n}} P_{A}(\rho).
\end{equation}
\end{definition}

A matrix $\rho_A$  maximizing  $P_{A}$ is called an equilibrium density operator for $A$.\\

Question:  Is there $\mathcal{L}$ such that $\rho_A$ is stationary for  $\mathcal{L}$? The converse in Theorem 3.1 in \cite{CaMa} may be useful.

The expression of entropy we found in Proposition \ref{h_eigen_dependence} suggests us to look at the matrices of the form $\xi =\rho^{1/2}$, where $\rho$ is a density. In order to study the problem of who maximizes $P_{A}$, we then define

\begin{equation}\label{mamai1}\Xi_{n} = \{\xi\ge 0:\xi^{2}\in D_{n}\},
\end{equation}

\noindent the set of square roots of density operators and
\begin{equation}\label{mamai2}\p_{A}(\xi) = h_{1/2}(\xi) + \tr(A\xi^{2}),
\end{equation}

\noindent where $h_{1/2}(\xi) := 2\tr(\xi)^{2} - 2n = h(\xi^{2})$ by Proposition \ref{h_eigen_dependence}. Notice that $\p_A(\xi)=P_A(\xi^2)$.

\begin{proposition}
If $\xi$ maximizes the functional $\p_A$, then there exists a $\kappa$   satisfying

\begin{equation}\label{mamai} 2 \kappa \xi = 4\tr(\xi)I + A\xi + \xi A.
\end{equation}

\end{proposition}
\begin{proof}
We will use Lagrange Multipliers. Let $g: M^{n} \to \C$, $g(\zeta) = \tr(\zeta^{2}) - 1$. The maximal $\xi$ then satisfies, for all $h\in M_{n}(\C)$ and some $\kappa \in \mathbb{R}$,

\begin{equation}\label{lagrange}
  \begin{cases}
    \D\p_{A}(\xi)(h) = \kappa \D g(\xi)(h) \\ g(\xi^{2}) = 0,
  \end{cases}
\end{equation}

\noindent with $\D\p_{A}(\xi)(h) = 4\tr(\xi)\tr(h) + \tr(A\{\xi, h\})$ and
$Dg(\xi)(h) = 2\tr(h\xi)$.\\

Taking $h=\outerp{e_{j}}{e_{i}}$ we have

$$  2\kappa\xi_{ij} = 4\tr(\xi)\delta_{i,j} + \{A,\xi\}_{ij}.$$

Since the above is true for every $i,j$, it corresponds to the coordinate equations of the matrix equation
\begin{equation}\label{Dhejei}
 2 \kappa \xi = 4\tr(\xi)I + A\xi + \xi A.
 \end{equation}
\end{proof}

Note that if $\xi =\rho^{1/2}_A $, then it follows from \eqref{mamai}
\begin{equation}\label{mamai3}  \kappa \rho_A\,=\, 2\tr(\rho^{1/2}_A)\, \rho^{1/2}_A +\frac{1}{2}\, ( A \,\rho_A + \rho^{1/2}_A A \rho^{1/2}_A).
\end{equation}

Indeed,
\begin{equation}\label{mamai3}  2 \kappa \rho_A = 2 \kappa \xi \, \xi= 4\tr(\xi)\, \xi + A\xi \xi + \xi A \xi= 4\tr(\xi)\, \xi + \rho_A + \xi A \xi.
\end{equation}

\begin{proposition}\label{contas}
  If $\xi$ maximizes $\p_{A}$, then the following statements are true
  \begin{enumerate}
      \item $\tr(A\xi) = (\kappa - 2n)\tr(\xi)$;
      \item $P_{A} (\rho_A)=P_{A}(\xi^{2}) = \p_{A}(\xi) = \kappa - 2n$;
  \end{enumerate}
\end{proposition}

\begin{proof} In (\ref{lagrange}) take $h=\Id$ and $h=\xi$, respectively. \end{proof}

\medskip

\begin{remark}\label{diagonalizable}
The problem of finding the maximizing density $\rho=\xi^2$ for a general Hermitian $A$ can be reduced to the study of the diagonal case. In fact, since $A$ is hermitian, it is diagonalizable. We have $UAU^*=\Lambda$, the later a diagonal matrix, for some unitary matrix $U$. Multiplying on the left by $U$ and on the right by $U^*$ in the above matrix equation gives us
$$ 2 \kappa \ U \xi U^* = 4\tr(\xi)UU^* + UAU^*\ U\xi U^* + U \xi U^* \ U AU^*$$
$$ \iff 2 \kappa \eta = 4\tr(\eta) + \Lambda \eta + \eta \Lambda,$$

\noindent where $\eta:=U\xi U^*$. We arrive at a particular version of (\ref{Dhejei}) on which the matrix is diagonal.

\end{remark}

\begin{theorem}\label{Adiag=>xidiag}
  If $A=\Diag(a_{1},\ldots,a_{n})$, the $\xi$ that maximizes $\p_{A}$ is also diagonal, with
  $$\xi_{ii}= \frac{c}{\kappa - a_i},$$
  \noindent where $\kappa$ is given implicitly on the data $a_1,...,a_n$ and $c$ is such that $\tr(\xi^2)=1$. Consequently, the density that maximizes the pressure $P_A$ is
$$ \rho_A= \frac{1}{\tr\, (\rho_A)}\begin{pmatrix}
\frac{1}{(\kappa - a_1)^2} &   & \\
 & \frac{1}{(\kappa - a_2)^2}  & \\
& & \ddots & \\
& & & \frac{1}{(\kappa - a_n)^2}
\end{pmatrix}.$$
\end{theorem}

\begin{proof}
For $A$ diagonal, the expression (\ref{Dhejei}) becomes
\begin{equation}\label{Dhejei_diag}
  2\kappa\xi_{ij} = 4\tr(\xi)\delta_{i,j} + \xi_{ij}(a_{i}+a_{j}).
\end{equation}

If we take $i=j$ in the expression above,

\begin{equation}\label{Dhejei_diag_ii0}
  \kappa\xi_{ii} = 2\tr(\xi) + \xi_{ii}a_{i}
\iff \newline
  (\kappa - a_i )\xi_{ii} = 2\tr(\xi)
\end{equation}

We know that $\tr(\xi) >0$, because $\tr(\xi)=0$ leads to $\xi=0$. Then $\xi_{ii}\ne 0$ and $\kappa > a_{i}$ for all $i$. So,

\begin{equation}\label{Dhejei_diag_ii}
  \frac{1}{\kappa - a_{i}} = \frac{\xi_{ii}}{2\tr(\xi)}
\end{equation}

and

\begin{equation}\label{Dhejei_diag_ii_sum}
  \sum_{i=1}^{n}\frac{1}{\kappa - a_{i}} = \frac{1}{2}.
\end{equation}

We need to find $\kappa$ to completely characterize the maximal $\xi$. Suppose that $a_1 = \max_i a_{i}$. Notice that $f(x)=\sum_{i=1}^{n} \frac{1}{x-a_i},$ for $x\neq a_i$ has a vertical asymptote at $a_1$, $\lim_{x\to a_1^{+}} f(x) = \infty$, and it decreases to $\lim_{x \to \infty} f(x) = 0$. By the intermediate value theorem, we have a $\kappa > a_1$ s.t. $f(\kappa)=1/2$.

Alternatively, one can find it as one of the roots of the following polynomial, which is the expression (\ref{Dhejei_diag_ii}) rewritten.
\begin{equation}\label{lambda_polinom}
  \frac{1}{2}\det(\kappa\Id - A) -
  \sum_{i=1}^{n}\det\Big(\kappa\Id - A + (a_{i}-\kappa+1)\outerp{e_{i}}{e_{i}}\Big) = 0.
\end{equation}

There is just one root that is bigger than all $a_i$, therefore it is $\kappa$.\\

Finally we prove the elements out of the diagonal are null. For $i\ne j$, the expression (\ref{Dhejei_diag_ii_sum}) gives us $2\kappa\xi_{ij} = \xi_{ij}(a_{i}+a_{j})$, or equivalently,
$$\xi_{ij}(2\kappa-(a_{i}+a_{j}))=0.$$  We know that $\kappa > a_i, \forall i$. Thus $2\kappa > a_i + a_j$. This leaves us with $\xi_{ij}=0$.\\

To conclude the pressure problem, we write
\begin{equation}
\xi = \Diag(\xi_{11},...,\xi_{nn}), \quad \xi_{ii}= \frac{c}{\kappa - a_i},
\end{equation}
where $c$ is the constant that makes $\tr(\xi ^2)=1$, i.e.,
$$ c = \left( \sum_{i=1}^{n} \frac{1}{(\kappa - a_i)^2} \right) ^{-1/2}.$$

This way, given $a_1,...,a_n$, we find $\kappa$, then $c$ and finally $\xi$.

\end{proof}

\begin{corollary} If $A$ is diagonalizable, with $UAU^*$ diagonal, then the maximizing density $\rho_A$ for $P_A$ is such that $U\rho_A U^*$ is diagonal.
\end{corollary}

\begin{proof}

If $A$ was not diagonal at first, we proceed as in Remark \ref{diagonalizable} and use the last theorem to find a maximal $\eta=U\xi U^*$ which is diagonal. Then $\eta^2=U\xi^2U^*=U\rho_AU^*$ is diagonal.
\end{proof}

\medskip

\begin{remark}
  Using (\ref{Dhejei_diag_ii}), we know that if $A = a_{1}\Id$ then $\xi_{ii} = \xi_{11}$ for all $i$. Since $A$ is diagonal and $\tr\,(\xi^2)= \sum_{i}\xi_{ii}^2=n\xi_{11}^2=1$,
  it follows that $\xi = \frac{1}{\sqrt{n}}\Id$ is the only $\xi$ that
  maximizes $\p_{A}$.
\end{remark}

\medskip

\begin{example}
Let
$$ A= \begin{pmatrix}
0 & 1 & 0 \\
1 & 0 & 0 \\
0 & 0 & 2 \\
\end{pmatrix}.$$
What is the equilibrium density $\rho_A$, i.e., the density that maximizes $P_A$?
Let's diagonalize $A$.
$$ U=\frac{1}{\sqrt{2}}\begin{pmatrix}
1 & -1 & 0 \\
1 & 1 & 0 \\
0 & 0 & \sqrt{2} \\
\end{pmatrix}, \quad UAU^* = \begin{pmatrix}
-1 & 0 & 0 \\
0 & 1 & 0 \\
0 & 0 & 2 \\
\end{pmatrix}.
$$
\end{example}

Now we apply Theorem \ref{Adiag=>xidiag}. $\kappa$ satifies
$$\dfrac{1}{\kappa+1} + \dfrac{1}{\kappa-1} + \dfrac{1}{k-2} = \dfrac{1}{2},$$
\noindent and $\kappa > -1,1,2$. We get $\kappa \approx 6.902$. Thus, the maximizing density  for $UAU^*$ is
$$ \rho_{UAU^*}= \frac{1}{\tr\, (\rho_{UAU^*})}\begin{pmatrix}
\frac{1}{(6.902 + 1)^2} &   & \\
 & \frac{1}{(6.902 - 1)^2}  & \\
& & \frac{1}{(6.902 - 2)^2}
\end{pmatrix}$$

$$ = \begin{pmatrix}
0.186 &   & \\
 & 0.332  & \\
 &  & 0.482
 \end{pmatrix}.$$

Finally,
$$\rho_A = U^*\rho_{UAU^*}U $$

$$ = \begin{pmatrix}
0.259 & 0.073  & 0 \\
0.073 & 0.259  & 0 \\
0 & 0  & 0.482
 \end{pmatrix}.$$

 Notice that $$P(\rho_A)=h(\rho_A)+\tr(A\rho_A)$$
 $$=2\,\tr(\rho_A^{1/2})^2-2n+\tr(UAU^*U\rho_AU^*)$$
 $$=2\,\tr((U\rho_AU^*)^{1/2})^2-2n+\tr(UAU^*U\rho_AU^*)$$
 $$ = 2\, (\sqrt{0.186}+\sqrt{0.332}+\sqrt{0.482})^2-6+(-1\cdot 0.186+ 1\cdot 0.332 + 2\cdot 0.482) $$
 $$\approx 0.902 \approx \kappa - 6.$$
in accordance with Proposition \ref{contas}.

\section{The pressure $ P_A$ as an eigenvalue problem} \label{trans}

Consider the linear operator $\mathfrak{L}_A$
\begin{equation}\label{Dheje35}
 \xi \to  \mathfrak{L}_A(\xi)=   2\tr(\xi)I + \frac{1}{2} (A\xi + \xi A).
 \end{equation}

 Suppose  $\rho_A$  is an equilibrium density operator for  the selfadjoint matrix $A$.

From \eqref{mamai} we get that $\xi =\rho^{1/2}_A $ is an eigenmatrix for the linear operator $\mathfrak{L}_A$ associated to the eigenvalue $\kappa$, that is
\begin{equation}\label{mamai6}
\mathfrak{L}_A (\rho^{1/2}_A) = \kappa\, \rho^{1/2}_A
\end{equation}

From item 2. in Theorem \ref{contas} we get that
$P_{A} (\rho_A) = \kappa - 2n.$

In this way, the  equilibrium density operator is related to an eigenvalue problem in a similar fashion as in classical Thermodynamic Formalism.

The equilibrium matrix $\rho_A$ satisfies
\begin{equation}\label{mamai5}  \kappa \rho_A\,=\, 2\tr(\rho^{1/2}_A)\, \rho^{1/2}_A +\frac{1}{2}\, ( A \,\rho_A + \rho^{1/2}_A A \rho^{1/2}_A),
\end{equation}
but this is not exactly a linear relation.

\medskip

\section{A connection between $h(\rho)$ and $I(\rho)$} \label{Kaka}

We will present a connection between the concepts of $h(\rho)$ and $-I(\rho)$.

Recall that $\mathfrak{1}=Id/n$ satisfies the detailed balance condition, which is the quantum equivalent of reversibility.

We are going to establish a connection of the Laplacian-entropy of  Definition \ref{ent} with the one in (5.18) of \cite{Str}.
The notion of Radon-Nikodym derivative is not clear in the quantum setting, but we can consider a natural  analogy in our reasoning  and we write $\frac{d\nu}{d\mu} = A$, if
$$ \tr(\nu U) = \tr(\mu AU).$$
\noindent This corresponds of  writing $A$ in the form $A=\mu^{-1} \nu $.
When looking at (5.18), $L$ is simmetric in $L^2(\mu)$, and our operator
$\mathcal{L}_0$ satisfies d.b.c. for $\mathfrak{1}$. Therefore, here we will
address the computation of  the corresponding expression
$\frac{d\nu}{d\mathfrak{1}}$. Then, it is natural to consider  $A= \mathfrak{1}^{-1}\nu = n \nu$. Therefore, (5.18) in our setting corresponds to
$$I(\nu) = - \int (n\nu)^{1/2}\mathcal{L}_0((n\nu)^{1/2}) \, d\mathfrak{1} $$
$$ = - n\tr(\mathfrak{1} \, \nu^{1/2}\mathcal{L}_0(\nu^{1/2}) )$$
$$ = - \tr(\nu^{1/2}\mathcal{L}_0(\nu^{1/2})),$$
\noindent and then, $-h(\nu)$ come up.

\begin{remark}\label{simplifica} Recall that for an $A=\sum_{kl} a_{kl} \ket{k}\bra{l}$,
$$ \mathcal{L}_0(A) = \sum_{i,j=1}^n (V_{ij}^*[A,V_{ij}] + [V_{ij},A]V_{ij}^*), \, \text{ for } V_{ij}=\ket{i}\bra{j}$$
$$ = \sum_{i,j=1}^n V_{ji}AV_{ij}-V_{ji}V_{ij}A + V_{ij}AV_{ji}-AV_{ij}V_{ji} $$
$$ = 2 \sum_{i,j=1}^n V_{ji}AV_{ij} - \sum_{i,j=1}^n \ket{j}\bra{i}\ket{i}\bra{j} A - \sum_{i,j=1}^n A \ket{i}\bra{j}\ket{j}\bra{i} $$
$$ = 2 \sum_{i,j=1}^n V_{ji}AV_{ij} - n\,\sum_{j=1}^n \ket{j}\bra{j} A - n\,\sum_{i=1}^n A \ket{i}\bra{i} $$
$$ = 2 \sum_{i,j=1}^n V_{ji}AV_{ij} - 2n A$$
$$ = 2 \sum_{i,j=1}^n \ket{j}\bra{i} A \ket{i}\bra{j} - 2n A$$
$$ = 2 \sum_{i,j,k,l=1}^n a_{kl} \ket{j}\bra{i}  \ket{k}\bra{l} \ket{i}\bra{j} - 2n A$$
$$ = 2 \sum_{i,j=1}^n a_{ii} \ket{j}\bra{j} - 2n A$$
$$ = 2 \sum_{i=1}^n a_{ii} \sum_{j=1}^n \ket{j}\bra{j} - 2n A$$
$$ = 2 \tr\, (A)I - 2n A.$$
\end{remark}

The next step is to write the entropy as an  infimum. We will show that:

\begin{theorem}\label{ent_as_inf} Given the density matrix $\rho$
$$ h(\rho) =  \inf_{A>0} \tr\,(\rho \,A^{-1}\, \mathcal{L}_0(A)).$$
\end{theorem}

For the proof, we will need the following result:

\begin{lemma}\label{inflemma} In the space $M_n$ of matrices $n\times n$, is true that
$$ \inf_{B>0} \tr\,(BU)\,\tr\,(UB^{-1}) = \tr\,(U)^2.$$
\end{lemma}

\begin{proof}
Let $B>0$ be a general positive matrix. Let $\ket{i}$ be the orthonormal basis of eigenvectors of $B$. Then, we can write $B=\sum_{i=1}^{n} b_i \, \ket{i}\bra{i}$, and in this basis, $U$ can be written as $U=\sum_{j,k=1}^{n}\, u_{jk} \, \ket{j}\bra{k}$. Thus,
$$ BU = \sum_{ijk}\, b_i u_{jk} \, \ket{i}\bra{i}\ket{j}\bra{k} = \sum_{ik} b_i u_{ik}\,  \ket{i}\bra{k}$$
$$ \Longrightarrow \tr(BU)= \sum_{i=1}^{n} b_i u_{ii}.$$

$$ UB^{-1} = \sum_{ijk} \frac{u_{jk}}{b_i}\, \ket{j}\bra{k}\ket{i}\bra{i} = \sum_{ij} \frac{u_{ji}}{b_i}\, \ket{j}\bra{i}$$
$$ \Longrightarrow \tr(UB^{-1})= \sum_{i=1}^{n} \frac{u_{ii}}{b_i}.$$

\noindent Therefore,

$$ \tr(UB)\,\tr(BU^{-1}) = \sum_{i,j=1}^{n} \frac{b_j}{b_i} \, u_{ii}u_{jj} = \frac{1}{2} \, \sum_{i,j=1}^{n} \left(\frac{b_j}{b_i}+\frac{b_i}{b_j}\right)\, u_{ii}u_{jj}$$
$$ \ge \sum_{i,j=1}^{n} u_{ii}u_{jj} = \left(\sum_{i=1}^{n} u_{ii}\right)^2=\tr(U)^2.$$

\noindent Notice that we used the fact that $x+1/x \ge 2$, for all $x>0$. By now, we have the lower bound $\tr(U)^2$. To finish, notice that $B=Id$ achieves this bound, so we conclude that
$$ \inf_{B>0} \tr(BU)\,\tr(UB^{-1}) = \tr(U)^2.$$
\end{proof}

Now we proceed to prove the theorem.
\begin{proof}
Using Remark \ref{simplifica}, we have
$$ A^{-1}\mathcal{L}_0(A) = 2 A^{-1} \tr(A) - 2n I$$
$$ \Rightarrow \tr(\rho A^{-1}\mathcal{L}_0(A)) = 2\, \tr(\rho A^{-1}).\tr(A)-2n\,\tr(\rho)$$
$$= 2 \,\tr(\rho A^{-1})\tr(A) - 2n.$$

Writing $A$ in the form $A=B\rho^{1/2}$, it will not change the  infimum, which will be now taken over $B>0$. This means
$$ \inf_{A>0} 2\, \tr(\rho A^{-1})\tr(A) - 2n = 2 \,\inf_{B>0}\tr(\rho^{1/2} B^{-1})\tr(B\rho^{1/2}) - 2n  $$
$$ = 2\,\tr(\rho^{1/2})^2 - 2n = h(\rho).$$
\noindent As the  infimum was computed by Lemma \ref{inflemma} we proved the claim.
\end{proof}

\section{From quantum to classical} \label{cla}

\begin{definition} Given $\sigma$ and an infinitesimal generator $\mathcal{L}$ of the form \eqref{gre},  we say that the matrix $Q$ is the matrix associated to $\mathcal{L}$, if $Q$ is   $n\times n$ real matrix with entries
$Q_{i,j} = \tr \, [F_{i,i}  \mathcal{L} F_{j,j}]$,
where $F_{i,i} =\ket{\eta_{i} }\bra{ \eta_{i}}$.
\end{definition}

This matrix is line sum zero with positive  values outside  the diagonal (see \cite{CaMa}). The matrix  $Q^\dag$, the transpose of $Q$, has a stationay   eigenvector probability $\overrightarrow{\sigma} \in (0,1)^n$ associated to the eigenvalue $0$.
\begin{lemma} Given $l,k$, the entry   $Q_{l,k} = \tr \,[F_{l,l}  \, \mathcal{L} F_{k,k}]$ is given by
\begin{equation}\label{boro} Q_{l,k} = 2\, e^{-w_{k,l}/2} -  2 \delta_{l,k} \, \sum_{i=1}^n e^{-w_{i,l}/2}.
\end{equation}
\end{lemma}

\begin{proof}
Indeed, when, $A=V_{k,k}= \outerp{\eta_{k}}{\eta_{k}} $ we get
$$  V_{i,j}^* [ A, V_{i,j}] =  \outerp{\eta_{j}}{\eta_{i}} \, \big[\, A \,,\outerp{\eta_{i}}{\eta_{j}}\,\big]=$$
$$ \outerp{\eta_{j}}{\eta_{i}} \,\, \big(\, \,\outerp{\eta_{k}}{\eta_{k}} \,\,\outerp{\eta_{i}}{\eta_{j}} - \outerp{\eta_{i}}{\eta_{j}} \,\, \outerp{\eta_{k}}{\eta_{k}} \,\big)=$$
$$ \delta_{i,k}\, \outerp{\eta_{j}}{\eta_{j}}
  - \delta_{j,k} \, \outerp{\eta_{j}}{\eta_{k}}$$

  Moreover,  when $A=V_{k,k}= \outerp{\eta_{k}}{\eta_{k}} $
$$   [ V_{i,j}^*, A]  \, V_{i,j}= \big[\,  \outerp{\eta_{j}}{\eta_{i}} ,\, A \,\big]\,\,\outerp{\eta_{i}}{\eta_{j}}=$$
$$ \,\, \big(\, \,\outerp{\eta_{j}}{\eta_{i}} \,\,\outerp{\eta_{k}}{\eta_{k}} - \,\outerp{\eta_{k}}{\eta_{k}} \,\,\outerp{\eta_{j}}{\eta_{i}} \,\big)\,\,\outerp{\eta_{i}}{\eta_{j}} =$$
$$ \delta_{i,k}\, \outerp{\eta_{j}}{\eta_{j}}
  - \delta_{j,k} \, \outerp{\eta_{k}}{\eta_{j}}.$$

  Then,  when $A=V_{k,k}= \outerp{\eta_{k}}{\eta_{k}} $
  $$ e^{-w_{i,j}/2} \, \big( V_{i,j}^* [ A, V_{i,j}]  + [V_{i,j}^*, A]  \, V_{i,j} \big)=  $$
  $$ e^{-w_{i,j}/2} \, \big(\,   2 \delta_{i,k}   \outerp{\eta_{j}}{\eta_{j}}   -  \delta_{j,k} \, \outerp{\eta_{j}}{\eta_{k}}   -  \delta_{j,k} \, \outerp{\eta_{k}}{\eta_{j}}  \,\big)$$
  $$ = e^{-w_{i,j}/2} \, \big(\,   2 \delta_{i,k}   \outerp{\eta_{j}}{\eta_{j}}   - 2 \delta_{j,k} \, \outerp{\eta_{k}}{\eta_{k}}  \,\big), \, $$
  and finally
$$ \mathcal{L}( A) = \sum_{i,j=1}^n   e^{-w_{i,j}/2} \, ( V_{i,j}^* [ A, V_{i,j}]  + V_{i,j}^*, A]  \, V_{i,j} )=$$
$$ \sum_{i=1}^n    \sum_{j=1}^n  e^{-w_{i,j}/2} \, (\,   2 \delta_{i,k}   \outerp{\eta_{j}}{\eta_{j}}   - 2 \delta_{j,k} \, \outerp{\eta_{k}}{\eta_{k}}\,)= $$
$$ 2 \,\sum_{j=1}^n   e^{-w_{k,j}/2}  \, \outerp{\eta_{j}}{\eta_{j}}   -  2\,\sum_{i=1}^n      e^{-w_{i,k}/2}    \, \outerp{\eta_{k}}{\eta_{k}}  $$

Therefore, when $A=V_{k,k}= \outerp{\eta_{k}}{\eta_{k}} $, given $l$
$$  \outerp{\eta_{l}}{\eta_{l}} \,\,\mathcal{L}( A) =   2\, \outerp{\eta_{l}}{\eta_{l}} \sum_{j=1}^n   e^{-w_{k,j}/2}  \, \outerp{\eta_{j}}{\eta_{j}}   -  2\, \outerp{\eta_{l}}{\eta_{l}} \sum_{i=1}^n      e^{-w_{i,k}/2}    \, \outerp{\eta_{k}}{\eta_{k}}  $$

$$ =  2\, \sum_{j=1}^n   e^{-w_{k,j}/2}  \outerp{\eta_{l}}{\eta_{l}} \outerp{\eta_{j}}{\eta_{j}}   -  2\, \sum_{i=1}^n      e^{-w_{i,k}/2}  \outerp{\eta_{l}}{\eta_{l}} \outerp{\eta_{k}}{\eta_{k}}  $$
$$ =  2\, e^{-w_{k,l}/2}  \outerp{\eta_{l}}{\eta_{l}} -  2 \delta_{l,k} \, \sum_{i=1}^n      e^{-w_{i,l}/2}  \outerp{\eta_{l}}{\eta_{l}}.$$

From this,

\begin{equation} \label{asl} Q_{lk} = 2\, e^{-w_{k,l}/2} -  2 \delta_{l,k} \, \sum_{i=1}^n e^{-w_{i,l}/2}
\end{equation}
$$ = \left\{
        \begin{array}{ll}
            2\, e^{-w_{k,l}/2} =2 \, e^{(\lambda_l- \lambda_k)/2}& \quad \text{if } l\neq k \\
            \\
           2\, e^{-w_{l,l}/2} -  2 \, \sum_{i=1}^n e^{-w_{i,l}/2} = 2 - 2 \sum_{i=1}^n e^{(\lambda_l- \lambda_i)/2} & \quad \text{if } l=k \\
        \end{array}
    \right.
$$

\end{proof}

Notice that:
$$ \sum_{k=1}^n Q_{lk} = Q_{ll} + \sum_{k: \, k\neq l} Q_{lk} $$
$$ = 2\, e^{-w_{l,l}/2} -  2 \, \sum_{i=1}^n e^{-w_{i,l}/2} + 2\, \sum_{k: \, k\neq l} e^{-w_{k,l}/2} $$
$$ = -  2 \, \sum_{i=1}^n e^{-w_{i,l}/2} + 2\, \sum_{k=1}^n e^{-w_{k,l}/2} = 0.$$

\smallskip

Note that the expression \eqref{boro} for the matrix $Q$ depends on the eigenvalues $e^{-\lambda_i}$, $i \in\{1,2..,n\}$, and not the specific eigenfunctions  $\eta_{i}$, $i \in\{1,2..,n\}$, of $\sigma$. This means that many density matrices $\sigma$ can determine the same matrix $Q$.

\smallskip
Theorem  4.2 in \cite{CaMa} claims:
\begin{theorem} Assume that $\mathcal{L}$ is  of the form \eqref{gre} for $\sigma$. The matrix
$Q$, given  by
$Q_{i,j} = $ Tr $[F_{i,i}  \mathcal{L} F_{j,j}]$ is line sum zero. The invariant probability for the classical continuous time Markov chain with infinitesimal generator $Q$ is
\begin{equation} \label{klr2}\overrightarrow{\sigma}=(\sigma_1, \sigma_2,..., \sigma_n)=(Tr[\sigma F_{1,1}],Tr[\sigma F_{2,2}],... Tr[\sigma F_{n,n}]).
\end{equation}
The classical detailed balance condition
\begin{equation} \label{klr3} \sigma_i Q_{i,k} = \sigma_k Q_{k,i}
\end{equation}
is   satisfied.

Consider the Chapman-Kolmogorov linear differential equation on $\overrightarrow{\rho}(t) =(\rho_1(t),\rho_2(t),...,\rho_n(t)) \in \mathbb{R}^n,$
\begin{equation} \label{klr} \frac{d}{dt} \rho_l (t) = \sum_{k=1}^n (\rho_k(t) Q_{k,l} - \rho_k(t) Q_{l,k}).
\end{equation}

This is equivalent to
\begin{equation} \label{klr14} \overrightarrow{\rho}(t) = e^{ t Q^\dag} (\overrightarrow{\rho}(0)).
\end{equation}

The occupation time probability in $\{1,2..,n\}$ of the continuous time Markov Chain is described by
$\overrightarrow{\rho}(t)$.

 $\overrightarrow{\rho}(t) $ satisfies \eqref{klr}, if and only if, the quantum continuous time  evolution $\rho(t)$ in
 $\mathcal{A}$  satisfies
\begin{equation} \label{klr1}  \rho(t) = \sum_{k=1}^n \frac{\rho_k (t)}{Tr (F_{k,k} )}\, F_{k,k}.
\end{equation}

\end{theorem}

Remember that from \eqref{tyty}
 \begin{equation} \label{klr10}\sigma = \sum_{k=1}^n  e^{-\, \lambda_k} \outerp{\eta_k}{\eta_k}.
 \end{equation}

 Then, from \eqref{klr2}, given $j$
 \begin{equation} \label{klr11}\sigma_j = \tr \,[\sigma\, \outerp{\eta_j}{\eta_j}\,]=\tr\,\Big[  \sum_{k=1}^n  e^{-\, \lambda_k} \outerp{\eta_k}{\eta_k} \, \outerp{\eta_j}{ \eta_j}\,\Big]=e^{- \lambda_j}.
 \end{equation}

  Expression \eqref{klr3} means for $k \neq l$
  \begin{equation} \label{klr12} e^{-\lambda_k} \,   e^{\lambda_k/2 -\lambda_l/2 }=   e^{-\lambda_k/2 -\lambda_l/2 } =  e^{-\lambda_l}\, e^{\lambda_l/2 -\lambda_k/2 }.
  \end{equation}

  \medskip

  From \eqref{asl} and \eqref{klr2} we get
  \begin{equation} \label{klr12} \overrightarrow{\sigma} Q=(\sigma_1, \sigma_2,..., \sigma_n) Q=(0,0,...,0).
  \end{equation}

\begin{example}\label{quantum to classical - example sigma diagonal}
Let $$\sigma=\begin{pmatrix}
\frac{1}{2} &0   &0\\
0 & \frac{1}{3}  &0 \\
0& 0& \frac{1}{6}
\end{pmatrix}.$$
Then $$h=-\log \, \sigma = \begin{pmatrix}
\log 2 &  0 &0 \\
 0& \log 3  & 0\\
0& 0& \log 6
\end{pmatrix},$$
\noindent so $\lambda_1=\log 2$, $\lambda_2=\log 3$ and $\lambda_3=\log 6$. The $Q$ matrix given by the expression (\ref{boro}) has entries

$$Q_{12}=2 e^{(\log 2 - \log 3)/2}=2\left(\frac{2}{3}\right)^{1/2}=2\frac{\sqrt 2}{\sqrt 3}$$
$$Q_{13}=2 e^{(\log 2 - \log 6)/2}=2\left(\frac{2}{6}\right)^{1/2}=2\frac{1}{\sqrt 3}$$
$$Q_{11} = 2\left(-\frac{\sqrt 2 }{\sqrt 3}- \frac{1}{\sqrt 3}\right)$$

$$Q_{21}=2 e^{(\log 3 - \log 2)/2}=2\left(\frac{3}{2}\right)^{1/2}=2\frac{\sqrt 3}{\sqrt 2}$$
$$Q_{23}=2 e^{(\log 3 - \log 6)/2}=2\left(\frac{3}{6}\right)^{1/2}=2\frac{\sqrt 3}{\sqrt 6}$$
$$Q_{22} = 2(-\frac{\sqrt 3}{\sqrt 2}- \frac{1}{\frac{\sqrt 3}{\sqrt 6}})$$

$$Q_{31}=2 e^{(\log 6 - \log 2)/2}=2\left(\frac{6}{2}\right)^{1/2}=2\sqrt 3$$
$$Q_{32}=2 e^{(\log 6 - \log 3)/2}=2\left(\frac{6}{3}\right)^{1/2}=2\sqrt 2$$
$$Q_{33} = 2(-\sqrt 3-\sqrt 2)$$

Thus
$$ Q = 2 \begin{pmatrix}
-\frac{\sqrt 2 }{\sqrt 3}- \frac{1}{\sqrt 3} & \frac{\sqrt 2}{\sqrt 3} & \frac{1}{\sqrt 3} \vspace{6pt} \\  \vspace{6pt}
\frac{\sqrt 3}{\sqrt 2}  & -\frac{\sqrt 3}{\sqrt 2}- \frac{\sqrt 3}{\sqrt 6} & \frac{\sqrt 3}{\sqrt 6}\\
\sqrt 3 & \sqrt 2 & -\sqrt 3-\sqrt 2
\end{pmatrix}.$$

We should have that $\vec{\sigma}=(\frac{1}{2},\frac{1}{3},\frac{1}{6})$ is the invariant vector. In fact,

\begin{align*}
 \frac{1}{2}(\vec{\sigma}Q)_1 &= -\frac{(\sqrt 2+1)}{2 \sqrt 3}+\frac{\sqrt 3}{3\sqrt 2} + \frac{\sqrt 3}{6} \\
&= \frac{-\sqrt 3 \sqrt 2 -\sqrt 3 + \sqrt 2 \sqrt 3 + \sqrt 3}{6}=0.
\end{align*}

\begin{align*}
\frac{1}{2}(\vec{\sigma}Q)_2 &= \frac{\sqrt 2}{2\sqrt 3} - \frac{(1+\sqrt 3)}{3 \sqrt 2} + \frac{\sqrt 2}{6} \\
 &= \frac{\sqrt 2 \sqrt 3 -\sqrt 2 - \sqrt 2 \sqrt 3 + \sqrt 2}{6}=0.
\end{align*}

\begin{align*}
 \frac{1}{2}(\vec{\sigma}Q)_3 &= \frac{1}{2\sqrt 3}  + \frac{1}{3 \sqrt 2} - \frac{(\sqrt 2 + \sqrt 3)}{6} \\
 &= \frac{\sqrt 3 +\sqrt 2 - \sqrt 2 - \sqrt 3}{6}=0.
\end{align*}

\end{example}

\begin{example}\label{quantum to classical - example sigma complex}
Let $$\sigma=\begin{pmatrix}
\dfrac{1}{4}  & 0            & \dfrac{i}{8}       \\
0             & \dfrac{1}{2} & 0                  \\
-\dfrac{i}{8} & 0            & \dfrac{1}{4}
\end{pmatrix}.$$

The eigenvalues of $\sigma$ are $\frac{1}{8}$, $\frac{3}{8}$ and $\frac{1}{2}$.
Then $h=-\log \, \sigma$ has eigenvalues
$\lambda_1=\log{8}$,
$\lambda_2=\log{\dfrac{8}{3}}$ and
$\lambda_3=\log{2}$. So, the $Q$ matrix given by the expression (\ref{boro}) has entries

$$Q_{12}=2 e^{(\log{8} - \log{\frac{8}{3}})/2}=2\left(3\right)^{1/2}=2\sqrt 3.$$
$$Q_{13}=2 e^{(\log{8} - \log{2})/2}=2\left(\frac{8}{2}\right)^{1/2}=4.$$
$$Q_{11} = -2\left(\sqrt 3 + 2\right).$$

$$Q_{21}=2 e^{(\log{\frac{8}{3}} - \log{8})/2}=2\left(\frac{1}{3}\right)^{1/2}=\frac{2}{\sqrt 3}.$$
$$Q_{23}=2 e^{(\log{\frac{8}{3}} - \log{2})/2}=2\left(\frac{4}{3}\right)^{1/2}=\frac{4}{\sqrt 3}.$$
$$Q_{22} = -\frac{6}{\sqrt{3}}.$$

$$Q_{31}=2 e^{(\log{2} - \log{8})/2}=2\left(\frac{1}{4}\right)^{1/2}= 1.$$
$$Q_{32}=2 e^{(\log{2} - \log{\frac{8}{3}})/2}=2\left(\frac{3}{4}\right)^{1/2}=\sqrt 3.$$
$$Q_{33} = -\left(1 + \sqrt 3\right).$$

Thus

$$ Q = \begin{pmatrix}
-2\left(\sqrt 3 + 2\right) && 2\sqrt 3             && 4                           \\\\

\dfrac{2}{\sqrt 3}         && -\dfrac{6}{\sqrt{3}} && \dfrac{4}{\sqrt 3}          \\\\

        1                  && \sqrt 3              && -\left(1 + \sqrt 3\right)
\end{pmatrix}.$$

We should have that $\vec{\sigma}=(\frac{1}{8},\frac{3}{8},\frac{1}{2})$ is the invariant vector. In fact,

\begin{align*}
(\vec{\sigma}Q)_1 &=
    \frac{1}{8}\left(-2\sqrt 3 - 4\right) +
    \frac{3}{8}\frac{2}{\sqrt 3} +
    \frac{1}{2} \\
&= -\frac{\sqrt 3}{4} - \frac{1}{2} + \frac{\sqrt 3}{4} + \frac{1}{2}
=0.
\end{align*}

\begin{align*}
(\vec{\sigma}Q)_2 &=
    \frac{\sqrt 3}{4} -
    \frac{3\sqrt 3}{4} +
    \frac{2\sqrt 3}{4} = 0.
\end{align*}

\begin{align*}
(\vec{\sigma}Q)_3 &=
    \frac{1}{2} +
    \frac{\sqrt 3}{2} -
    \frac{1}{2}\left(1 + \sqrt 3\right) = 0.
\end{align*}

\end{example}

Related results are described in (3) and (4) in \cite{Leeuw}.

\medskip

{\center IME -  UFRGS, Av. Bento Gonvalves 9500 - 91.500 Porto Alegre, Brazil}

\smallskip

Email of J. E. Brasil is jaderebrasil@gmail.com

Email of J. Knorst is  jojoknorst@gmail.com

Email of A. O. Lopes is arturoscar.lopes@gmail.com


\smallskip



\begin{thebibliography}{99}

\bibitem{AlFa}
R. Alicki and M. Fannes. Quantum Dynamical Systems. Oxford University
Press  (2000).


\bibitem{BEL} A. Baraviera, R. Exel and A. O. Lopes,
A Ruelle Operator for continuous time Markov Chains, Sao Paulo Journal of Mathematical Sciences, vol 4 n. 1, pp 1-16 (2010)

\bibitem{BKL1} J. E. Brasil, and J. Knorst and A. O.  Lopes, Thermodynamic Formalism for Quantum Channels: Entropy, Pressure, Gibbs Channels and generic properties,   Communications in Contemporary Mathematics.   Volume 25, Issue 04, 2150090 (2023)

\bibitem{BKL2} J. E. Brasil, and J. Knorst and A. O.  Lopes, Lyapunov exponents for Quantum Channels: an entropy formula and generic properties,
 Journal of Dynamical Systems and Geometric Theories, 21(2) 155-187 (2021)


\bibitem{CaMa} E. A. Carlen and J. Maas,
Gradient flow and entropy inequalities for  quantum Markov semigroups with detailed balance,
Journal of Functional Analysis, 273.  1810--1869 (2017) 

\bibitem{Chang}
M.-H. Chang, Quantum Stochastics. Cambridge University Press, (2015)



\bibitem{Cipri} F. Cipriani, Dirichlet Forms and Markovian Semigroups on Standard
Forms of von Neumann Algebras,
Journal of functional analysis 147, 259--300 (1997)


\bibitem{FR} F. Fagnola and R. Rebolledo, 
From classical to quantum entropy production, In  Quantum Probability and Infinite Dimensional Analysis, 245-261 (2010)

\bibitem{JPW}  V. Jaksic, C.-A. Pillet and M. Westrich,
Entropic fluctuations of
quantum dynamical semigroups, J. Stat. Phys., 154(1-2), 153–187, (2014)



\bibitem{Kac}
M. Kac, Integration in function spaces and some of its applications, Acad Naz dei
Lincei Scuola Superiore Normale Superiore, Piza, Italy (1980).

\bibitem{KLMN} J. Knorst,  A. O. Lopes, G. Muller and A. Neumann, Thermodynamic Formalism on the Skhorohd space:  the continuous time Ruelle operator,  entropy, pressure, entropy production and  expansiveness.
Sao Paulo Journal of Math. Sciences, 18:1414-1446 (2024)


\bibitem{Leeuw} M. de Leeuw, C. Paletta and Balazs Pozsgay,
Constructing Integrable Lindblad Superoperators, arXiv (2021)



\bibitem{LNT} A. O. Lopes,  A. Neumann and Ph. Thieullen, A thermodynamic formalism for continuous time Markov chains with values on the Bernoulli Space: entropy, pressure and large deviations,
Journ. of Statist. Phys. Volume 152, Issue 5, Page 894-933 (2013)


\bibitem{LMMS}
A.~O. Lopes, J.~K. Mengue, J.~Mohr, and R.~R. Souza,
Entropy and variational principle for one-dimensional lattice systems
  with a general {\it a priori} probability: positive and zero temperature.
Ergodic Theory Dynam. Systems, 35(6):1925--1961 (2015)

\bibitem{PP}
W. Parry and M. Pollicott. Zeta functions and the periodic orbit strucuture of hyperbolic dynamics. Asterisque, Vol 187--188, pp 1--268 (1990).

\bibitem{Pillet} C. A. Pillet,
Quantum dynamical systems, In Open Quantum Systems I. The Hamiltonian Approach. S. Attal, A. Joye and C.-A. Pillet editors. Lecture Notes in Mathematics 1880.
Springer, Berlin (2006)


\bibitem{Str} D. W. Stroock, An Introduction to the Theory of Large Deviations, Springer Verlag (1984)

\bibitem{Wolf} M. M. Wolf, Quantum Channels and Operations - Guided Tour. 2010.

https://www-m5.ma.tum.de/foswiki/pub/M5/
Allgemeines/MichaelWolf/ QChannelLecture.pdf

\end{thebibliography}
\end{document}